\pdfoutput=1
\documentclass[10pt,pra,aps,superscriptaddress,nofootinbib,twocolumn,longbibliography]{revtex4-2}
\usepackage{amsmath,bbm}
\usepackage{amsthm}
\usepackage{amsmath}
\usepackage{latexsym}
\usepackage{amssymb}
\usepackage{float}
\usepackage{braket}
\usepackage{quantikz}
\usepackage{tikz}
\usepackage{graphicx}           
\usepackage{textcomp}
\usepackage{color}
\usepackage{xcolor}
\usepackage{mathpazo}
\usepackage{comment}
\usepackage{enumitem}
\usepackage{multirow}
\usepackage{setspace}
\usepackage[colorlinks=true,linkcolor=blue,citecolor=magenta,urlcolor=blue]{hyperref}
\usepackage{svg}
\usepackage[caption=false]{subfig}

\newcommand{\be}{\begin{equation}}
\newcommand{\ee}{\end{equation}}
\newcommand{\bea}{\begin{eqnarray}}
\newcommand{\eea}{\end{eqnarray}}
\def\squareforqed{\hbox{\rlap{$\sqcap$}$\sqcup$}}
\def\qed{\ifmmode\squareforqed\else{\unskip\nobreak\hfil
\penalty50\hskip1em\null\nobreak\hfil\squareforqed
\parfillskip=0pt\finalhyphendemerits=0\endgraf}\fi}
\def\endenv{\ifmmode\;\else{\unskip\nobreak\hfil
\penalty50\hskip1em\null\nobreak\hfil\;
\parfillskip=0pt\finalhyphendemerits=0\endgraf}\fi}

\newcommand{\tr}{\text{Tr}}
\newcommand{\I}{\mathbbm{1}}

\newcommand{\la}{\langle}
\newcommand{\ra}{\rangle}

\makeatletter
\newtheorem*{rep@theorem}{\rep@title}
\newcommand{\newreptheorem}[2]{%
\newenvironment{rep#1}[1]{%
 \def\rep@title{#2 \ref{##1}}%
 \begin{rep@theorem}}%
 {\end{rep@theorem}}}
\makeatother

\newtheorem{thm}{Theorem}%[section]
\newreptheorem{thm}{Theorem}

\newtheorem{obs}{Observation}

\begin{document}
%%%%%%%%%%%%%%%%%%%%%%%%%%%%%%%%%%%%%%%%%%%%%%%%%%%%%%%%%%%%%%%%%%%

\title{Global versus Local Discrimination of Locally Implementable Multipartite Unitaries}

%%%%%%%%%%%%%%%%%%%%%%%%%%%%%%%%%%%%%%%%%%%%%%%%%%%%%%%%%%%%%%%%%%%

\author{Satyaki Manna}
\affiliation{School of Physics, Indian Institute of Science Education and Research Thiruvananthapuram, Kerala 695551, India}
\affiliation{Department of Physics, School of Basic Sciences, Indian Institute of Technology Bhubaneswar, Odisha 752050, India}
\author{Sneha Suresh}
\affiliation{School of Physics, Indian Institute of Science Education and Research Thiruvananthapuram, Kerala 695551, India}
\author{Anandamay Das Bhowmik}
\affiliation{School of Physics, Indian Institute of Science Education and Research Thiruvananthapuram, Kerala 695551, India}
\affiliation{S. N. Bose National Centre for Basic Sciences, Block JD, Sector III, Salt Lake, Kolkata 700 106, India}
\author{Debashis Saha}
\affiliation{Department of Physics, School of Basic Sciences, Indian Institute of Technology Bhubaneswar, Odisha 752050, India}

%%%%%%%%%%%%%%%%%%%%%%%%%%%%%%%%%%%%%%%%%%%%%%%%%%%%%%%%%%%%%%%%%%%

\begin{abstract}
We study single-shot distinguishability of locally implementable multipartite unitaries under Local Operations and Classical Communication (LOCC) and global operations. As unitary discrimination depends on both the choice of probing states and the measurements on the evolved states, we classify LOCC and global distinguishability into two categories: adaptive strategies, where probing states are chosen based on measurement outcomes from other subsystems, and restricted strategies, where probing states remain fixed. Our findings uncover three surprising features in the bipartite setting and establish new structural limits for unitary discrimination: (i) Certain pairs of unitaries are globally distinguishable with restricted strategies but indistinguishable under LOCC, even with adaptive strategies. (ii) There exist sets of four unitaries that are distinguishable via LOCC, yet remain globally indistinguishable with restricted strategies. (iii) Some sets of unitaries are globally indistinguishable when probed with separable states, but become distinguishable via LOCC.
\end{abstract}

\maketitle

%%%%%%%%%%%%%%%%%%%%%%%%%%%%%%%%%%%%%%%%%%%%%%%%%%%%%%%%%%%%%%%%%%% 

\section{Introduction}
One of the most fascinating features of quantum theory is captured by the phenomenon known as \textit{nonlocality without entanglement} \cite{Benett}. This arises when local quantum operations assisted by classical communication fail to perfectly distinguish a set of multipartite quantum states that are orthogonal as well as separable, and thus locally preparable. This demonstrates a profound quantum signature that the whole system can possess properties not discernible from its parts, even in the absence of entanglement. Over the years, this phenomenon has attracted significant attention for its foundational implications and its relevance to quantum information processing \cite{Walgate,Walgate2,nathanson2005distinguishing,Bandyopadhyay_2011,Bhattacharya,halder, Cohen2,PhysRevA.85.042319,ghosh_s,zhen,supic,ha2021quantum,He2024,Niset,hayashi,watrous05,Horodecki}.

While this effect in quantum states has been extensively studied, a fundamental question arises: Does a similar feature exist for the discrimination of quantum processes? Specifically, one can compare the distinguishability of locally implementable multipartite quantum channels under two paradigms: global operations, where a single user has access to all subsystems of an unknown multipartite process, and local operations with classical communication (LOCC), where each spatially separated party can access only their local subsystem.

Although channel distinguishability is a well-explored field \cite{AYuKitaev_1997, Piani,Dariano2,Sacchi,Harrow,William,watrous2008, PhysRevA.111.022221,GMauroDAriano_2002,entangledstatesusefulsingle,acin,acin2,duan, Duan2, Duan3, Ziman10022010,njp2021,entangledstatesusefulsingle,acin,acin2,duan, Duan2, Ziman10022010,bavaresco,bavaresco2022unitary}, the comparative study of global versus LOCC strategies for discriminating quantum channels remains largely uncharted \cite{cohen,Heng,chitambar}. Prior work on unitary operations within the LOCC framework has focused primarily on unitaries that are not locally implementable \cite{Heng,Duan3}. The problem becomes significantly more interesting and subtle when we consider product unitary operations, that is, when they are locally implementable. Discrimination of quantum channels involves two controllable components: the choice of the probing state on which the unknown channel acts, and the measurement performed on the evolved state to extract information about the unitary. Accordingly, discrimination protocols can be classified into two types: restricted strategies, where the probing states remain fixed throughout the protocol, and adaptive strategies, where probing states are chosen based on measurement outcomes from other subsystems. This additional layer of complexity arises only in the case of locally implementable multipartite channels.

The primary goal of this work is to study the distinguishability of product unitaries under both global and LOCC operations. We formulate the problem by classifying both global and local strategies into restricted and adaptive types. After establishing generic relationships between these strategies, we present three key findings for bipartite unitaries. First, certain pairs of qubit unitaries are globally distinguishable using restricted strategies but remain indistinguishable under LOCC, even with adaptive strategies.
Second, there exist sets of four qutrit unitaries that are distinguishable via LOCC but indistinguishable globally when restricted strategies are used.
And third, some sets of qubit unitaries are globally indistinguishable when probed with separable states, yet become distinguishable via LOCC. In addition, we show that for bipartite qubit unitaries, adaptive and restricted LOCC strategies are equivalent. We also find that LOCC distinguishability can be asymmetric, meaning a set of unitaries may be distinguishable if one party initiates the protocol but not if the other does.
These findings reveal unexpected separations between global and local discrimination of product unitaries, exposing new fundamental limitations on what can be achieved with LOCC and global strategies.

We begin by revisiting the problem of distinguishing a set of unitary operations, and then formalize the notion of global and LOCC protocols.

%The paper is structured as follows. In the next section, we revisit the problem of distinguishing a set of unitary operations. Subsequently, we formulate the notion of global and LOCC protocols for the distinguishability of unitaries. Then we move into the section where all the main results are given. In conclusion, we summarize the key findings and discuss several open problems with potential future research directions.

\section{Distinguishability of Unitary Operations}\label{sec2}
Unitary operator $(U)$ is a linear operator $U:H\rightarrow H$ on a Hilbert space $H$ that satisfies $U^\dagger U=UU^\dagger=\I$.
We consider a priori known set of $m$ unitaries $\{U_i\}_{i=1}^m$ acting on a $d$-dimensional quantum state. We assume that all the unitaries are sampled from a given probability distribution. To distinguish the unitaries, the unitary device is fed with a known quantum state, single or entangled, and the device carries out one of $m$ unitary operations. After this process, the device gives an evolved state as the output. Therefore, any measurement of the evolved state can be performed. Let us describe this measurement by a set of POVM (Positive Operator Valued Measurement) elements $\{N_{b}\}$, where $b\in\{1,\cdots,m\}$. The protocol is successful in distinguishing the unitaries if $b$ is the same as $i$. Any classical post-processing of outcome $b$ can be included in the measurement $\{N_{b}\}$. Distinguishability of a set of unitaries eventually reduces to the distinguishability of evolved states. If the probe is single system $\rho$, this problem is equivalent to discriminating the set of states $\{U_i\rho U_i^\dagger\}_{i=1}^m$. In general, one considers a bipartite probe state $\rho_{AB}$, and the task reduces to distinguishing the set of states $\{(U_i\otimes\I)\rho_{AB}(U_i^\dagger\otimes\I)\}_{i=1}^m$. Needless to say, 
if $\rho_{AB}$ is separable (i.e., a product probe), the strategy offers no advantage over a single-system probe and performs equivalently \cite{entangledstatesusefulsingle}. Several previous works proved that a single system is sufficient for the discrimination of two unitaries, that means entangled probes do not perform better than single systems in this case \cite{entangledstatesusefulsingle,GMauroDAriano_2002}. In fact, there exists a necessary and sufficient condition for two unitaries to be distinguishable. Two unitaries $U_1$ and $U_2$ are distinguishable if $\min|con\{e^{\mathbbm{i}\theta_j}\}|= 0$, where  $\{e^{\mathbbm{i}\theta_j}\}_j$
are the eigenvalues of $U_1^\dagger U_2$, $con\{e^{\mathbbm{i}\theta_j}\}$ denotes the set of complex numbers that can be written as the convex combinations of $\{e^{\mathbbm{i}\theta_j}\}_j$ and $\min|\cdot|$ denotes the minimum norm over all those complex numbers. For the discrimination among more than two unitaries, there are instances where entangled probes are strictly better than single system probes.
%
% $(i)$ When probing state is single system $\rho$, distinguishability of the unitaries $\{U_x\}_x$, sampled from probability distribution $\{p_x\}_x$, is following:
% \bea\label{D_S}
% \mathcal{D}_S\left[\{U_x\}_x,\{p_x\}_x\right]
% &=& max_\rho \mathcal{DS}[\{U_x\rho U_x^\dagger\}_x,\{p_x\}_x].\nonumber\\
% \eea
% $(ii)$ Similarly, if the probing state is an entangled one, i.e., $\rho_{AB}$, the distinguishability of the same set of unitaries will be,
% \bea\label{D_E}
% &&\mathcal{D}_E\left[\{U_x\}_x,\{p_x\}_x\right]\nonumber\\
% &=& max_{\rho_{AB}} \mathcal{DS}[\{(U_x\otimes\I)\rho_{AB}(U_x^\dagger\otimes\I)\}_x,\{p_x\}_x].
% \eea
% $\mathcal{DS}[\{\rho_k\}_k,\{q_k\}_k]$ denotes the distinguishability of the set of states $\{\rho_k\}_k$, sampled from the probability distribution $\{q_k\}_k$. 
More details of the distinguishability of unitaries can be found in \cite{entangledstatesusefulsingle,GMauroDAriano_2002}.

\subsection{Protocols for Global distinguishability}
In this section, we present general protocols for global distinguishability of $m$ number of $n$-partite unitaries $\{U_i\}_{i=1}^m =\{\otimes_{k=1}^n U_i^{(k)}\}_{i=1}^m$. In general, the dimension of the unitaries can be different for different parties.

\subsubsection{Global Distinguishability with restricted strategy (GDR)}
In the global restricted strategy for distinguishing a set of unitary operations, the multipartite unitary is applied simultaneously across all subsystems, meaning the user cannot select and implement the operation on individual subsystems. Thus, without loss of generality, the user prepares an arbitrary composite quantum state single or entangled $\rho_{AB}$ on the composite Hilbert space $\mathcal{H}_A \otimes \mathcal{H}_B$, where $\mathcal{H}_A = \bigotimes_{k=1}^n \mathcal{H}_k $ represents the $n$-partite subsystem on which the unitaries act. The unknown unitary $U_i$ is applied to $A$ subpart of $\rho_{AB}$ resulting in an evolved state $\sigma_{AB}^i = (U_i \otimes \I_B) \rho_{AB} (U_i^{\dagger} \otimes \I_B)$. The party then performs arbitrary quantum operations in order to distinguish between these evolved states $\{\sigma^{i}_{AB}\}_{i=1}^m$, thereby identifying the unknown unitary operation from the given set of possibilities.

\subsubsection{Global Distinguishability with adaptive strategy (GDA)}\label{global_adap}
In this case, the user can sequentially choose different subsystems of the $n$-partite unitary to discriminate them. Suppose the user first selects an arbitrary subset of parties, $\mathbbm{S}_1 \subseteq [n]$, where $[n] \equiv \{1,\cdots,n\}$, consisting of $n_1$ subsystems with $n_1\leq n$. The user then prepares an initial probing state $\rho_{AB}$, which may in general be entangled across the composite Hilbert space $\mathcal{H}_A \otimes \mathcal{H}_B$, where $\mathcal{H}_A = \bigotimes_{k\in\mathbbm{S}_1} \mathcal{H}_k$ represents the $n_1$-partite subsystem on which the unitaries act. The evolved state is then $\sigma_{AB}^i = (\otimes_{k\in\mathbbm{S}_1}U_i^{(k)} \otimes \I_B) \rho_{AB} (\otimes_{k\in\mathbbm{S}_1} U_i^{(k)} \otimes \I_B)^\dagger$. The user performs a measurement to discriminate among these transformed states. If perfect discrimination is achieved, the protocol terminates successfully. Otherwise, the user employs a suitable measurement to eliminate some of the unitaries. Based on this partial information, the user selects another subset $\mathbbm{S}_2 \subseteq [n]\setminus \mathbbm{S}_1$ of size $n_2$, prepares an appropriate probing state, and attempts to distinguish the corresponding $n_2$-partite unitaries. If successful, the protocol ends; otherwise, further elimination is performed. The process is repeated until the user exhausts all the subsystems on which the multipartite unitary acts.
There are many possible subsets the user may choose, and each sequence of subsets defines a partition of $[n]$. Moreover, for any given partition, the protocol also depends on the order in which the subsets are chosen. Consequently, the user considers all the strategies arising from all possible permutations of subsets across all partitions and applies any that perfectly distinguish the unitaries. 
%the user adopts an optimal strategy by maximizing over all possible permutations of subsets across all partitions. 
 
\subsection{Protocols for LOCC distinguishability}
In contrast, a local protocol involves $n$ spatially separated parties, indexed by $k = \{1,\cdots,n\} $, such that each party has access to one partite from the $n$-partite unidentified unitary $U_i = U^{(1)}_{i}\otimes U^{(2)}_{i}\otimes\cdots\otimes U^{(n)}_{i} = \bigotimes_{k=1}^{n} U^{(k)}_{i} $. In the protocol, party $k$ prepares the probing state $\rho_{A_kB_k}$ on $\mathcal{H}_{A_k} \otimes \mathcal{H}_{B_k}$, where $\mathcal{H}_{A_k}$ denotes the Hilbert space on which $U^{(k)}_{i}$ acts.

\subsubsection{LOCC distinguishability with restricted strategy (LDR)}
In the restricted class of local strategies, each party independently prepares a state $\rho_{A_kB_k}$ and lets the subsystem $A_k$ evolve under the unknown unitary. After a unitary $U_{i}$ acts on the respective subsystems, the evolved state for the $k$-th party is given by $\sigma^{i}_{A_kB_k} = (U_i^{(k)} \otimes \I_B) \rho_{A_kB_k} ( U_i^{(k)} \otimes \I_B)^\dagger$. Collectively, the parties are then left with one of $m$ possible $n$-partite states, $\{\bigotimes_{k=1}^{n} \sigma^{i}_{A_kB_k}  \}_{i=1}^m$, corresponding to the given set of unitaries $\{U_i\}$. Their task, therefore, boils down to distinguishing these evolved states using only local quantum operations on their respective systems, combined with classical communication among themselves. Note that, the parties choose the best possible initial probing states.

\subsubsection{LOCC distinguishability with adaptive strategy (LDA)}
In general, local strategies can be more intricate than the restricted ones described above. In an adaptive LOCC protocol, $k$-th party prepares its probing state $\rho_{A_kB_k}$ based on classical information received from other parties who have already performed measurements on their probed systems. Such a protocol involves many possible sequences in which the parties participate.
Say, the protocol begins with an initial party, say $k$, who prepares a state $\rho_{A_kB_k}$ and applies $U^{(k)}_{i}$ from the unknown unitary $U_{i}$. On the resulting state, the party performs a quantum measurement and obtains an outcome. Depending on the outcome, the party eliminates certain possibilities for $U^{(k)}_{i}$ and consequently for $U_i$ as well, and conveys the remaining possibilities for $U_i$ to another party, say $k'$. Based on this updated information, the $k'$-th party prepares a suitable state $\rho_{A_{k'}B_{k'}}$, applies the respective local unitary, and further eliminates specific possibilities among the remaining candidates for $U_i$.
This process continues sequentially, with each party reducing the set of possible unitaries through local quantum operations and passing the updated information along. Importantly, the choice of the next party depends on both the classical outcome obtained by the previous party and the information received in the earlier rounds. Thus, communication between the parties need not be strictly one-way, and the probing state and measurement by the subsequent party depend on the accumulated information. In this manner, the $n$ separated parties iteratively narrow down the set of possible unitaries until a unique unitary is identified. It is worth noting that for the special case of discriminating between two unitaries (\textit{i.e.}, $m=2$), the adaptive strategy provides no advantage over the restricted one.

%This process proceeds sequentially, with each party in the order reducing the set of possibilities for $U^{(k)}_{i}$ through their operations and passing the updated information about the remaining possibilities for $U_{i}$ to the next party according to a strategy. It is important to note that, depedning on the classical output one party obtains and the information he recerived before, he chooses the next party who will excute the next operation, thereby the communcaiiton may not be one-way between the parties. moreover, the next party's probing state and measurement depend on this information.In this way, $n$ separated parties iteratively eliminate the candidates for $U_i$, each contributing to the reduction of the global set of possibilities until they can identify one unitary. However, it is evident that when distinguishing between only two quantum operations, \textit{ i.e.}, for $m = 2$, the adaptive strategy offers no additional advantage over the restricted one. 

Although adaptive strategies in the multipartite setting can become highly complex with many possible configurations, in the following, we focus on the bipartite case, which is comparatively simpler, to present all the main results. The schematic representation of the circuits for the four types of protocols to distinguish bipartite unitaries is provided in FIG. \ref{fig}. It is useful to introduce a simplified subclass of local adaptive strategy, namely \textit{one-way LDA}. Here, the protocol involves a single round of local quantum operations on each subsystem, which is conditioned on prior classical communication.
\begin{widetext}

\begin{figure}
    \centering
    %\begin{subfigure}[b]{0.4\textwidth}
    %\centering
    \subfloat[Global distinguishability with restricted strategy (GDR)]{
    \begin{minipage}{0.48\linewidth}
\centering
   \begin{tikzpicture}[scale=1,
  decoration={brace,mirror,amplitude=6pt}] % load brace decoration
    
    % Wires before M
    \draw (0,0) -- (2.8,0);     % Wire 1
    \draw (0,-0.8) -- (2.8,-0.8); % Wire 2
    \draw (0,-1.6) -- (2.8,-1.6); % Wire 3
    
    % Input label
    \node[left] at (-0.5,-0.5) {$\ket{\psi}_{123}$};
    \node[left] at (0,0) {$1$};
    \node[left] at (0,-0.8) {$2$};
    \node[left] at (0,-1.6) {$3$};

    % Big left bracket for wires 1–3
    \draw[decorate] (-0.35,0.2) -- (-0.35,-1.8);

    % Bigger U boxes
    \draw[fill=white] (1.5,-0.3) rectangle (2.5,0.3);     % U_i^(1)
    \node at (2.0,0) {$U_i^{(1)}$};

    \draw[fill=white] (1.5,-1.1) rectangle (2.5,-0.5);    % U_i^(2)
    \node at (2.0,-0.8) {$U_i^{(2)}$};
    
    % Measurement box spanning 3 wires
    \draw[fill=gray!20,rounded corners] (2.8,-1.8) rectangle (4.2,0.2);
    \node at (3.5,-0.8) {$M$};
    
    \draw (3,-0.5) arc[start angle=180,end angle=0,radius=0.5 cm];
      
    % Draw slanted arrow inside the half circle
    \draw[->, thick] (3.4,-0.5) -- (3.8,-0.2);
    \draw[double, -{Latex[length=2mm]}] (4.2,-0.8) -- (5.2,-0.8);
    \node[right] at (5.3,-0.8) {$z = i$};
    
\end{tikzpicture}
    %\caption{Global distinguishability with restricted strategy (GDR)}
    \end{minipage}
    \label{figa} }
    %\end{subfigure}
    \hfill
    %\begin{subfigure}[b]{0.54\textwidth}
    %\centering
    \subfloat[Global distinguishability with adaptive strategy (GDA)]{
    \begin{minipage}{0.48\linewidth}
\centering
    \begin{tikzpicture}[scale=1,
  decoration={brace,mirror,amplitude=6pt}] % load brace decoration
    
    % Horizontal wires
    \draw (0,0) -- (6.2,0);        % k
    \draw (0,-0.8) -- (6.2,-0.8);  % 3
    \draw (0,-2.1) -- (6.2,-2.1);  % k̄
    \draw (0,-2.6) -- (6.2,-2.6);  % 4
    
    % Wire labels
    \node[left] at (0,0) {$k$};
    \node[left] at (0,-0.8) {$3$};
    \node[left] at (0,-2.1) {$\bar{k}$};
    \node[left] at (0,-2.6) {$4$};
    \node[left] at (-0.35,-1.6) {$\ket{\psi}_{k3\bar{k}4}$};
    % U_i^{(k)} gate
    \draw[fill=white] (1,-0.3) rectangle (2,0.3);
    \node at (1.5,0) {$U_i^{(k)}$};

     \draw[decorate] (-0.35,0.2) -- (-0.35,-2.8);
    
    % Small M box
    \draw[fill=gray!20,rounded corners] (2.5,-0.95) rectangle (3.5,0.4);
    \node at (3,-0.4) {$M$};
    
    \draw (2.6,-0.1) arc[start angle=180,end angle=0,radius=0.4 cm];
      
      % Draw slanted arrow inside the half circle
      \draw[->, thick] (2.9,-0.1) -- (3.2,0.15);
    % Double classical arrow from M to V_a 
    \draw[double, -{Latex[length=2mm]}] (3,-0.95) -- (3,-1.8);
    \node at (3.15,-1.3) {$a$};
    % V_a box on wire 3
    \draw[fill=white] (2.5,-2.8) rectangle (3.5,-1.8);
    \node at (3,-2.3) {$V_a$};
    
    % U_i^{(1)} box on wire 3 (or appropriate wire)
    \draw[fill=white] (4,-2.4) rectangle (5,-1.7);
    \node at (4.5,-2.05) {$U_i^{(\bar{k})}$};
    
    % Big M box on the right (final measurement box)
    \draw[fill=gray!20,rounded corners] (5.2,-2.9) rectangle (6.4,0.3);
    \node at (5.8,-1.4) {$M$};
    
    \draw (5.3,-0.8) arc[start angle=180,end angle=0,radius=0.5 cm];
      
      % Draw slanted arrow inside the half circle
      \draw[->, thick] (5.7,-0.8) -- (6.1,-0.5);
    % Only one classical horizontal arrow from big M to z=i
    \draw[double, -{Latex[length=2mm]}] (6.4,-1.4) -- (7,-1.4);
    \node[right] at (7,-1.4) {$z = i$};
    
    \end{tikzpicture}
    %\caption{Global distinguishability with adaptive strategy (GDA)}
    \end{minipage}
    \label{figb} }
    %\end{subfigure}\\
    %\hfill
    \vspace{0.3cm}
    
   % \begin{subfigure}[b]{0.4\textwidth}
    %\centering
    \subfloat[LOCC distinguishability with restricted strategy (LDR)]{
    \begin{minipage}{0.48\linewidth}
\centering
    \begin{tikzpicture}[scale=1,
  decoration={brace,mirror,amplitude=6pt}] % load brace decoration
    
    % Wires before M
    \draw (0,0) -- (2.8,0);     % Wire 1
    \draw (0,-0.8) -- (2.8,-0.8); % Wire 2
    \draw (0,-1.6) -- (2.8,-1.6); % Wire 3
    \draw (0,-2.4) -- (2.8,-2.4); % Wire 4
    
    % Input label
    \node[left] at (-0.4,0) {$\ket{\psi}_{13}$};
    \node[left] at (-0.4,-1.7) {$\ket{\psi^{'}}_{24}$};
    \node[left] at (0,0) {$1$};
    \node[left] at (0,-0.8) {$2$};
    \node[left] at (0,-1.6) {$3$};
    \node[left] at (0,-2.4) {$4$};
    % Bigger U boxes
    \draw[fill=white] (1.1,-0.3) rectangle (2.1,0.3);     % U_i^(1)
    \node at (1.6,0) {$U_i^{(1)}$};
    
     \draw[decorate] (-0.35,0.17) -- (-0.35,-0.9);
     \draw[decorate] (-0.35,-1.5) -- (-0.35,-2.5);
    
    \draw[fill=white] (1.1,-1.9) rectangle (2.1,-1.3);    % U_i^(2)
    \node at (1.6,-1.6) {$U_i^{(2)}$};
    % vertical dashed arrow
    \draw[dashed,->] (2.6,0) -- (2.6,-2.8);
    \draw[dashed,-] (-0.4,-1.15) -- (4,-1.15);
    \node at (2.8,-3) {$\ket{\phi_{i}}_{1234}$};
    % Measurement box spanning 3 wires
    \draw[fill=gray!20,rounded corners] (2.8,-2.6) rectangle (4.2,0.2);
    \node at (3.5,-0.3) {LOCC};
    \node at (3.5,-0.8) {Dist};
    \node at (3.5,-1.2) {of};
    \node at (3.5,-1.6) {$\{\ket{\phi_{i}}\}_i$};
    
    % \draw (3,-0.5) arc[start angle=180,end angle=0,radius=0.5 cm];
      
    %   % Draw slanted arrow inside the half circle
    %   \draw[->, thick] (3.4,-0.5) -- (3.8,-0.2);
    
    %   % Draw slanted arrow inside the half circle
    %   \draw[->, thick] (3.4,-0.5) -- (3.8,-0.2);
      
    \draw[double, -{Latex[length=2mm]}] (4.2,-1.2) -- (5.2,-1.2);
    \node[right] at (5.3,-1.2) {$z = i$};
    
    \end{tikzpicture}
    %\caption{LOCC distinguishability with restricted strategy (LDR)}
    \end{minipage}
    \label{figc} }
    %\end{subfigure}
    \hfill
    %\begin{subfigure}[b]{0.54\textwidth}
    %\centering
    \subfloat[LOCC distinguishability with adaptive strategy (LDA)]{
    \begin{minipage}{0.48\linewidth}
\centering
    \begin{tikzpicture}[scale=1,
  decoration={brace,mirror,amplitude=6pt}] % load brace decoration
    
    % Horizontal wires
    \draw (0,0) -- (6.2,0);        % k
    \draw (0,-0.8) -- (6.2,-0.8);  % 3
    \draw (0,-2.1) -- (6.2,-2.1);  % k̄
    \draw (0,-2.6) -- (6.2,-2.6);  % 4
    
    % Wire labels
    \node[left] at (0,0) {$k$};
    \node[left] at (0,-0.8) {$3$};
    \node[left] at (0,-2.1) {$\bar{k}$};
    \node[left] at (0,-2.6) {$4$};
    \node[left] at (-0.45,-0.2) {$\ket{\psi}_{k3}$};
    \node[left] at (-0.45,-2.1) {$\ket{\psi^{'}}_{\overline{k}4}$};
    
    % U_i^{(k)} gate
    \draw[fill=white] (1,-0.3) rectangle (2,0.3);
    \node at (1.5,0) {$U_i^{(k)}$};

    \draw[decorate] (-0.35,0.17) -- (-0.35,-0.9);
    \draw[decorate] (-0.35,-1.9) -- (-0.35,-2.7);
    
    % Small M box
    \draw[fill=gray!20,rounded corners] (2.5,-0.95) rectangle (3.5,0.4);
    \node at (3,-0.4) {$M$};
    
    \draw (2.6,-0.1) arc[start angle=180,end angle=0,radius=0.4 cm];
      
      % Draw slanted arrow inside the half circle
      \draw[->, thick] (2.9,-0.1) -- (3.2,0.15);
    % Double classical arrow from M to V_a 
    \draw[double, -{Latex[length=2mm]}] (3,-0.95) -- (3,-1.8);
    \node at (3.15,-1.25) {$a$};
    % V_a box on wire 3
    \draw[fill=white] (2.5,-2.8) rectangle (3.5,-1.8);
    \node at (3,-2.3) {$V_a$};
    
    % U_i^{(1)} box on wire 3 (or appropriate wire)
    \draw[fill=white] (4,-2.4) rectangle (4.9,-1.7);
    \node at (4.5,-2.05) {$U_i^{(\bar{k})}$};
    
    % Big M box on the right (final measurement box)
        \draw[fill=gray!20,rounded corners] (5.2,-2.9) rectangle (6.5,0.3);
    \draw[dashed,->] (5,0) -- (5,-2.8);
    \draw[dashed,-] (-0.4,-1.4) -- (5.2,-1.4);
    \node at (5,-3.1) {$\ket{\phi_{i}}_{k3\overline{k}4}$};
    % Only one classical horizontal arrow from big M to z=i
    \draw[double, -{Latex[length=2mm]}] (6.5,-1.4) -- (7.1,-1.4);
    \node at (5.9,-0.3) {LOCC};
    \node at (5.9,-0.8) {Dist};
    \node at (5.9,-1.2) {of};
    \node at (5.9,-1.6) {$\{\ket{\phi_{i}}\}_i$};
     \node at (5.9,-2.1) {given};
     \node at (5.9,-2.4) {$a$};
    \node[right] at (7.2,-1.4) {$z = i$};
    
    \end{tikzpicture}
    %\caption{LOCC distinguishability with adaptive strategy (LDA)}
    \end{minipage}
    \label{figd}}
    %\end{subfigure}
    \caption{Circuit diagrams illustrating four different protocols for distinguishing a set of \textit{bipartite} product unitaries $U_i^{(1)}\otimes U_i^{(2)}$. Each unitary $U_i^{(k)}$ acts on subsystem $k$ of the input state. In the circuits, single lines represent quantum systems, while double lines denote classical outcomes obtained from a measurement $M$. The outcome $z$ is the guess for the unitary of index $i$, which should ideally be the same as $z$. Except for the unitaries $U_i^{(k)}$, which are given in advance, all other operations --- states, unitaries, and measurements are arbitrary and can be freely chosen. The label "LOCC Dist. of $\{\ket{\phi_i}\}_i$" refers to a LOCC protocol used to distinguish the states $\{\ket{\phi_i}\}_i$. The variables $(k,\bar{k})$ can be either $(1,2)$ or $(2,1)$, reflecting the two possible orderings in which $U^{(1)}$ and $U^{(2)}$ may act on after another in the adaptive strategy for bipartite unitaries. The horizontal dashed line in figures \ref{figc} and \ref{figd} denotes the spatial separation between the local parties. }
    \label{fig}
\end{figure}
\end{widetext}

Any restricted strategy can be read as an adaptive strategy in which all parties choose their probing state without implementing the information coming from the previous parties.
\begin{obs}
Distinguishability for any set of multipartite unitaries follows the ordering:
\begin{align}\label{obse_1}
GDR & \ \ \subseteq \ \ GDA \notag \\[0.4em]
&\hspace{-0.5 cm}\rotatebox[origin=c]{90}{$\subseteq$}  \hspace{1.2 cm} \rotatebox[origin=c]{90}{$\subseteq$} \notag \\[0.4em]
LDR & \ \ \subseteq \ \ LDA.
\end{align}
\end{obs}
% 
% \begin{obs}\label{obs_2}
%     If a set of bipartite unitaries is not perfectly distinguishable in one-way LDA, they are not perfectly distinguishable in one-way GDA either.
% \end{obs}
% The only difference in GDA and LDA protocols, in the bipartite setting, is that global measurement is allowed in GDA, but that is not the case for LDA. In LDA, any one party, say Alice, starts the protocol. If for a certain outcome of Alice's measurement, she could not eliminate a set of unitaries, that means the Kraus operator corresponding that outcome has non-zero inner product with those evolved states coming from non-eliminated unitaries acting on the optimum probe. This also implies that these evolved states are non-orthogonal to each other. So, distinguishing this non-eliminated unitaries, the other party, Bob, has to choose some probe such that the evolved states on Bob's side become orthogonal. In this case, the local measurement on Bob's side is sufficient to check if the evolved states are orthogonal as we already know the corresponding evolved states on Alice's side are non-orthogonal. Hence, a global measurement cannot distinguish these product states if Bob's states are not orthogonal. If those are orthogonal, local measurement is enough.

\begin{thm}\label{thm_1st}
    For any set of bipartite qubit unitaries, LDA is equivalent to LDR for perfect distinguishability.
\end{thm}
\begin{proof}
LDA becomes strictly better than LDR when the last party needs different probing states to distinguish a different set of unitaries. Theorem $4$ of \cite{entangledstatesusefulsingle} states that if there are some sets of distinguishable qubit unitaries, then every set is distinguishable by any common maximally entangled probe. Therefore, in the restricted strategy, the second party can always choose any maximally entangled state irrespective of the communication coming from the previous party. 
\end{proof}
\section{Main Results}
In this section, we produce some explicit examples of the set of unitaries to show the hierarchy of the global and LOCC distinguishability in the context of restricted and adaptive strategies. 
For the first study, we give an example of bipartite unitaries that are locally indistinguishable but globally distinguishable. Such a pair of unitaries in $\mathbbm{C}^2\otimes \mathbbm{C}^2$ is following: 
\be\label{unitary1}
\mathcal{U}_1=Q_1\otimes R_1, \ \ 
\mathcal{U}_2=Q_2\otimes R_2 ,
\ee
$ Q_1=\ket{0}\!\bra{0}+e^{-\mathbbm{i}\alpha}\ket{1}\!\bra{1},$ $  
 Q_2 = \ket{0}\!\bra{0}+e^{\mathbbm{i}\gamma}\ket{1}\!\bra{1},$ $R_1=\ket{0}\!\bra{0}+e^{-\mathbbm{i}\beta}\ket{1}\!\bra{1},$ $  R_2 = \ket{0}\!\bra{0}+e^{\mathbbm{i}\delta}\ket{1}\!\bra{1}$ and $\alpha,\beta,\gamma,\delta < \pi/2,$ $ \alpha+\beta+\gamma+\delta=\pi$.
\begin{thm}\label{thr5}
    The unitaries described by \eqref{unitary1} are globally distinguishable with restricted strategy but not LOCC distinguishable.
\end{thm}
\begin{proof}
Consider a GDR, in which one party has complete access to both parts of the unitaries $\mathcal{U}_1$ and $\mathcal{U}_2$. The party uses the Bell state $\ket{\phi^+}=\frac{1}{\sqrt{2}}(\ket{00}+\ket{11})$ as the probe. After the application of the unitaries, the possible evolved states are $\mathcal{U}_1\ket{\phi^+}=(Q_1\otimes R_1)\ket{\phi^+}=\frac{1}{\sqrt{2}}(\ket{00}+e^{-\mathbbm{i}(\alpha+\beta)}\ket{11})$ and $\mathcal{U}_2\ket{\phi^+}=(Q_2\otimes R_2)\ket{\phi^+}=\frac{1}{\sqrt{2}}(\ket{00}+e^{\mathbbm{i}(\gamma+\delta)}\ket{11})$. To distinguish these two unitaries, the party needs to discriminate between the above two transformed states. One can compute the modulus of their inner product as $|\la \phi^+|\mathcal{U}_1^\dagger \mathcal{U}_2|\phi^+\ra|=\frac12|1+e^{\mathbbm{i}(\alpha+\beta+\gamma+\delta)}|$. Since $\alpha+\beta+\gamma+\delta=\pi$, these states are orthogonal, and hence perfectly distinguishable.

We now analyze LOCC protocols for distinguishing $\mathcal{U}_1$ and $\mathcal{U}_2$. Each party has access to two local unitaries. Suppose Alice initiates the protocol. In a single round, Alice can eliminate at most one candidate unitary; eliminating both would not provide any useful information to Bob. Therefore, a two-outcome measurement suffices for Alice, where each outcome eliminates one possibility. This reduces the problem to distinguishing the pair of unitaries $Q_1$ and $Q_2$. We compute $Q_1^\dagger Q_2 = \ket{0}\!\bra{0} + e^{i(\alpha+\gamma)} \ket{1}\!\bra{1}$. As discussed in Sec.~\ref{sec2}, perfect distinguishability requires $\min |con\{1, e^{i(\alpha+\gamma)}\}| = 0$, which holds only if $e^{i(\alpha+\gamma)} = -1$. However, from the definitions of $\alpha$ and $\gamma$, this condition is not satisfied. Therefore, Alice cannot perfectly distinguish $Q_1$ and $Q_2$. Consequently, Bob must eventually distinguish between the two unitaries. A similar argument shows that $R_1$ and $R_2$ are also not perfectly distinguishable. In the case of two-way LOCC protocols, classical communication from Bob cannot enhance Alice's ability to distinguish her local unitaries, as Bob's operations cannot increase distinguishability beyond what is achievable initially.\\
If Bob instead initiates the protocol, the same limitation applies. Since $R_1$ and $R_2$ are indistinguishable, he cannot eliminate any unitary perfectly in the first step. Consequently, regardless of which party initiates the protocol, the set of unitaries is not perfectly distinguishable under either LDR or LDA strategies.
\end{proof}

We move into an example of a set of four bipartite unitaries that can be discriminated by the local adaptive strategy but not by the global restricted strategy. To construct such an example, one must go beyond qubit unitaries, as Theorem \ref{thm_1st} points out that for any set of bipartite qubit unitaries, local restricted and local adaptive strategies are equivalent. Let us take the following four qutrit unitaries,
\bea\label{new2}
& \mathcal{V}_1=\I_3\otimes \I_3, \ \
\mathcal{V}_2=\I_3\otimes \Omega , \nonumber\\
& \mathcal{V}_3=\Omega\otimes \I_3 , \ \
\mathcal{V}_4=\Omega\otimes \Gamma ,
\eea
where $\Omega= \ket{0}\!\bra{0}+e^{\mathbbm{i}2\pi/3}\ket{1}\!\bra{1}+e^{\mathbbm{i}4\pi/3}\ket{2}\!\bra{2},$ $\Gamma=\ket{0}\!\bra{0}+\ket{1}\!\bra{1}-\ket{2}\!\bra{2}$, and $\I_d$ denotes the identity operator on $\mathbbm{C}^d$.
% and $e^{\mathbbm{i}\theta_1}+e^{\mathbbm{i}\theta_2}\neq 0, \theta_1\neq\theta_2, 0<\theta_1,\theta_2<\pi$ and $\min |con\{e^{\mathbbm{i}\theta_1}, e^{\mathbbm{i}\theta_2}\}|= 0$.
%  $con\{e^{\mathbbm{i}\theta_1}, e^{\mathbbm{i}\theta_2}\}$ denotes the set of complex numbers that can be written as the convex combinations of $\{e^{\mathbbm{i}\theta_j}\}_{j=1}^2$ and $\min|.|$ denotes the minimum norm over all those complex numbers.
 \begin{thm}\label{thr3}
 The set of unitaries described at \eqref{new2} is distinguishable under LDA, but indistinguishable via GDR. 
 \end{thm}
\begin{proof}
    We first prove that the unitaries described by $\{\mathcal{V}_i\}_{i=1}^4$ are distinguishable in LDA.

For LDA, Alice starts the protocol by selecting the starting state $\ket{\Phi_s}=\frac{1}{\sqrt{3}}(\ket{0}+\ket{1}+\ket{2})$. After these unitaries act on this state, Alice will use the measurements having Krauss operators $\{\ket{\Phi_s}\bra{\Phi_s},\Omega\ket{\Phi_s}\bra{\Phi_s}\Omega^\dagger,\I-(\ket{\Phi_s}\bra{\Phi_s}+ \Omega\ket{\Phi_s}\bra{\Phi_s}\Omega^\dagger)\}$, where $\Omega\ket{\Phi_s}=\frac{1}{\sqrt{3}}(\ket{0}+e^{\mathbbm{i}2\pi/3}\ket{1}+e^{\mathbbm{i}4\pi/3}\ket{2})$. As a consequence, Bob has to distinguish between the different set of evolved states of his side depending on Alice's outcome, i.e, if the first outcome of Alice clicks, he has to distinguish between
$(\I_3,\Omega)$ and if it is the second outcome of Alice, he has to distinguish between $(\I_3,\Gamma)$ and there will be no occurrence of the third outcome. We can calculate that $\I_3^\dagger \Omega =\ket{0}\bra{0}+e^{\mathbbm{i}2\pi/3}\ket{1}\bra{1}+e^{\mathbbm{i}4\pi/3}\ket{2}\bra{2}$ and $\I_3^\dagger \Gamma = \ket{0}\bra{0}+\ket{1}\bra{1}-\ket{2}\bra{2}$. The necessary condition for $\I_3$ and $\Omega$ to be distinguishable is 
\be\label{B1}
\min|con\{1,e^{\mathbbm{i}2\pi/3},e^{\mathbbm{i}4\pi/3} \}|= 0,
\ee 
and similarly, for the perfect distinguishability of $\I_3$ and $\Gamma$, we need 
\be\label{B2}
\min|con\{1,1,-1\}|= 0.
\ee
For \eqref{B1}, it needs an equal convex mixture of three eigenvalues, and for \eqref{B2}, it requires an equal mixture of $1$ and $-1$. So parties successfully distinguish the unitaries in LDA.

For GDR to be $1$, we need to show that the evolved states after operation of the unitaries are orthogonal for the optimum probe $\ket{\psi_p}$.
The general probe will be a $4$-qutrit state, i.e., $\ket{\psi_p} = \sum_{x,y,z,w=0}^2\lambda_{xyzw}\ket{x}_A\ket{y}_B\ket{z}\ket{w}$. So we need to show that  $\langle\psi_p|\mathcal{V}_i^\dagger \mathcal{V}_j\otimes\I\otimes\I|\psi_p\rangle = 0$ for all $i\neq j$. Now we will check this orthogonality for the first pair of unitaries, i.e., $\mathcal{V}_1,\mathcal{V}_2$. With this probing state,
 \bea \label{R-rela}
 && \langle\psi_p|\mathcal{V}_1^\dagger \mathcal{V}_2\otimes\I\otimes\I|\psi_p\rangle\nonumber\\
 &=& \sum_{x,y,z,w}\lambda_{x,y,z,w}^*\lambda_{x,y',z,w}\bra{y}\I_3^\dagger \Omega\ket{y'}\nonumber\\
 &=& \sum_{x,y,z,w}|\lambda_{x,y,z,w}|^2 \bra{y}\I_3^\dagger \Omega\ket{y}\nonumber\\
 &=& 1(\mathcal{R}_1)+e^{\mathbbm{i}2\pi/3}(\mathcal{R}_2)+e^{\mathbbm{i}4\pi/3}(\mathcal{R}_3),
 \eea
  where,
  \bea\label{b5}
\mathcal{R}_1 &=& \sum_{x,z,w} |\lambda_{x,y=1,z,w}|^2, \ 
\mathcal{R}_2 =\sum_{x,z,w} |\lambda_{x,y=2,z,w}|^2,\nonumber\\
\mathcal{R}_3 &=& \sum_{x,z,w} |\lambda_{x,y=3,z,w}|^2.
  \eea
Similarly, we have
\bea\label{b6}
 \langle\psi_p|\mathcal{V}_3^\dagger \mathcal{V}_4\otimes\I\otimes\I|\psi_p\rangle
 &=& \sum_{x,y,z,w}|\lambda_{x,y,z,w}|^2\bra{y}\I_3^\dagger \Gamma\ket{y}\nonumber\\
 &=& \mathcal{R}_1+\mathcal{R}_2-\mathcal{R}_3.
 \eea
It is clear from \eqref{b5} that $\mathcal{R}_1,\mathcal{R}_2,\mathcal{R}_3$ are non-negative real numbers and sum to 1. Thus, to ensure that the expression \eqref{R-rela} is zero, we need $\mathcal{R}_1=\mathcal{R}_2=\mathcal{R}_3=1/3$. This condition does not make \eqref{b6} zero, that is, $1(\mathcal{R}_1)+ 1(\mathcal{R}_2)- 1(\mathcal{R}_3) =0$. That indicates $\ket{\psi_p}$ does not exist. Hence, global restricted strategy fails to discriminate the unitaries.
\end{proof}

The set of unitaries in \eqref{new2} is asymmetric. Numerical analysis further indicates that they are not distinguishable under one-way LDA when Bob initiates the protocol. This conclusion is obtained using semi-definite programming.

It is straightforward to see that Bob needs to eliminate at least two unitaries with each outcome of his measurement. Otherwise, Alice will end up with at least two identical states, which will be trivially indistinguishable irrespective of the restricted and adaptive strategy. For the elimination of two unitaries, Bob needs to remove one unitary among  $\I_3,\Omega$ and one unitary among $\I_3,\Gamma$, or he will eliminate three unitaries, which is a straightforward distinguishability condition. Thus, the favorable eliminations are $\{(\I_3,\Omega), (\I_3,\Gamma), (\Omega,\Gamma), (\I_3)\}$. In these possibilities, one can check that all three eliminations are also taken care of. Taking all the possibilities of Bob's elimination, using \textit{see-saw} method of semi-definite programming, we get the following optimization problem:
\bea
\mathcal{S}^{max}_{Q}
 &=& 1 -\min_{\rho_x,\{M_i\}_i} [ \tr( (\overline{\I_3}\rho_x \overline{\I_3}^\dagger + \overline{\Omega}\rho_x \overline{\Omega}^\dagger) M_1)\nonumber\\
 &&+ \tr ((\overline{\I_3}\rho_x \overline{\I_3}^\dagger + \overline{\Gamma}\rho_x \overline{\Gamma}^\dagger) M_2)\nonumber\\
 && +\tr (( \overline{\Omega}\rho_x \overline{\Omega}^\dagger + \overline{\Gamma}\rho_x \overline{\Gamma}^\dagger) M_3)\nonumber\\
&&+ \tr ((\overline{\I_3}\rho_x \overline{\I_3}^\dagger) M_4)\nonumber\\
\eea
\bea
 &\text{s.t.}& \rho_x\geqslant 0, \tr(\rho_x)=1, \rho_x\in\mathbbm{C}^d \nonumber\\
 && M_i\geqslant 0, \sum_{i=1}^4 M_i=\I, M_i \in\mathbbm{C}^d, d=9,\nonumber\\
 &\text{and}& \overline{\I_3}=\I_3\otimes\I_3,\overline{\Omega}=\Omega\otimes\I_3, \overline{\Gamma}=\Gamma\otimes\I_3.
\eea
  We find there does not exist any measurement $\{M_i\}_{i=1}^4$ such that $\mathcal{S}^{max}_{Q}=1$. Note that $M_i$ can take the value $0$, so if this optimization gives the highest value of $\mathcal{S}^{max}_{Q}$, we cannot be sure of Bob's success. Then we have to find all the nontrivial POVM elements and check that all the unitaries at Alice's side are exhaustively discriminated. But in this case, $\mathcal{S}^{max}_{Q}<1$, which means that Bob cannot execute any measurement such that the favorable options are eliminated at each outcome of the measurement. So, the unitaries are indistinguishable in LDA and Observation \ref{obse_1} implies that they will be indistinguishable in LDR also. Since the global restricted strategy remains unchanged, GDR cannot discriminate between unitaries. 

% \begin{proposition}
%     The LOCC distinguishability of the set of unitaries \eqref{new2} is asymmetric in nature. In particular, they are indistinguishable in one-way LDA if Bob starts the protocol.
% \end{proposition}

Finally, to demonstrate the usefulness of the entangled probing state, we identify the following set of LOCC distinguishable unitaries in $\mathbbm{C}^2\otimes \mathbbm{C}^2$ that becomes globally indistinguishable when only single (or separable) probing systems are allowed:
\bea\label{new}
& \mathcal{W}_1=\I_2\otimes \I_2, \ 
\mathcal{W}_2=Z\otimes X , \ 
\mathcal{W}_3=X\otimes H , \nonumber\\
& \mathcal{W}_4=X\otimes \overline{H}, \ \mathcal{W}_5=XZ\otimes H ,
\eea
with $Z=\ket{0}\!\bra{0}-\ket{1}\!\bra{1},$ $  X=\ket{1}\!\bra{0}+\ket{0}\!\bra{1}, $ $H=\ket{+}\!\bra{0}+\ket{-}\!\bra{1}$ and $\overline{H}=\ket{-}\!\bra{0}+\ket{+}\!\bra{1}$.
\begin{thm}
    The unitaries defined by \eqref{new} are indistinguishable under GDR and LDA if parties use only single systems probe, but distinguishable via LDR if the parties use entangled probes.
\end{thm}
\begin{proof}
First, we point out that the set of unitaries in \eqref{new} is not distinguishable under GDR when only single-system probes are allowed. In this case, the probing system has dimension four, while the unitaries generate five evolved states in $\mathbbm{C}^4$. Since at most four mutually orthogonal states can exist in a four-dimensional Hilbert space, these five states cannot be perfectly distinguished.

Second, we consider LDA with single-system probes, starting with the scenario in which Alice initiates the protocol. In order for Bob to complete the discrimination, Alice must eliminate at least three unitaries, leaving at most two candidates for Bob. This is necessary because Bob cannot distinguish more than two unitaries, as the corresponding evolved states lie in a two-dimensional subspace. The distinguishable pairs on Bob's side are $\{(\I_2, X), (H,\overline{H}), (\I_2,H), (X,\overline{H})\}$. Accordingly, the triples that Alice must eliminate are $\{(X,X,XZ), (\I_2,Z,XZ), (Z,X,XZ), (\I_2,X,XZ)\}$. To achieve such elimination using a single-system probe, Alice’s measurement must be chosen so that the corresponding POVM element is orthogonal to all three resulting states for any given triple. This is only possible if the POVM element is rank-one and the three states are identical up to a global phase. Among the listed triples, only $(X,X,XZ)$ satisfies this condition, with probe $\ket{0}$ or $\ket{1}$. Therefore, no measurement strategy allows Alice to eliminate the required triples in general. \\
It is also straightforward to see that multi-round LOCC protocols do not offer any advantage. Suppose Alice eliminates one or two unitaries in the first round, leaving at least two distinct unitaries. Since the elimination process necessarily involves rank-one POVM elements, the post-measurement states corresponding to the remaining unitaries become identical. Consequently, any subsequent operation by Bob, along with classical communication, cannot help Alice further discriminate among them. Hence, the protocol fails. \\
Similarly, if Bob starts the protocol, one can follow the same analysis as before. The triples Bob can eliminate are $(H,\overline{H},H)$ and $(\I_2,H,H)$. Eliminating these triples can make it possible to distinguish only $\mathcal{W}_1,\mathcal{W}_2$ and $\mathcal{W}_4$. So, irrespective of the starter, LDA remains unsuccessful in this distinguishing task. 

Finally, we show the third part of the proof, that is, the set of unitaries \eqref{new} is distinguishable under LDR with entangled probes. Let Alice and Bob implement the Bell state probe $\ket{\phi^+}=\frac{1}{\sqrt{2}}(\ket{00}+\ket{11})$ locally as the initial probing state. So, the evolved states are,
\bea\label{evolved_new}
(\I\otimes\I)_A\otimes(\I\otimes\I)_B\ket{\phi^+}^{\otimes 2}&=&\ket{\phi^+}\otimes\ket{\phi^+}, \nonumber\\
(Z\otimes\I)_A\otimes(X\otimes\I)_B\ket{\phi^+}^{\otimes 2}   &=& \ket{\phi^-}\otimes\ket{\psi^+}, \nonumber\\
    (X\otimes\I)_A\otimes(H\otimes\I)_B\ket{\phi^+}^{\otimes 2} &=& \ket{\psi^+}\otimes\ket{\tilde{\psi}^+}, \nonumber\\
   (X\otimes\I)_A\otimes(\overline{H}\otimes\I)_B\ket{\phi^+}^{\otimes 2} &=&\ket{\psi^+}\otimes\ket{\tilde{\psi}^-}, \nonumber\\  (XZ\otimes\I)_A\otimes(H\otimes\I)_B\ket{\phi^+}^{\otimes 2} &=& \ket{\psi^-}\otimes\ket{\tilde{\psi}^+},
    \eea
where $\ket{\phi^\pm}=\frac{1}{\sqrt{2}}(\ket{00}\pm \ket{11}),$ $ \ket{\psi^\pm }=\frac{1}{\sqrt{2}}(\ket{10}\pm\ket{01}), $   $\ket{\tilde{\psi}^\pm}=\frac{1}{\sqrt{2}}(\ket{\pm 0}+\ket{\mp 1})$.
Subsequently, Alice first measures on her subsystem in the Bell bases, that is, $\{\ket{\phi^+}, \ket{\phi^-}, \ket{\psi^+}, \ket{\psi^-}\}$. When the third outcome $(\ket{\psi^+})$ occurs, Eq.~\eqref{evolved_new} suggests that the state on Bob's side is either $\ket{\tilde{\psi}^+}$ or $\ket{\tilde{\psi}^-}$. Alice communicates her outcome to Bob, and then Bob measures his subsystem in $\{\ket{\tilde{\psi}^+},\ket{\tilde{\psi}^-}\}$ basis to distinguish between these two states, and consequently distinguishes between $\mathcal{W}_3$ and $\mathcal{W}_4$. For the other three outcomes of Alice, she can directly distinguish between $\mathcal{W}_1, \mathcal{W}_2$ and $\mathcal{W}_5$, which is evident from \eqref{evolved_new}.

If Bob starts the protocol, Bob measures in $\{\ket{\tilde{\psi}^+},\ket{\tilde{\psi}^-}\}$ basis. If the first outcome $(\ket{\tilde{\psi}^+})$ occurs, Bob communicates the outcome to Alice. On Alice's side, it reduces to three orthogonal Bell states, $\{\ket{\phi^-}, \ket{\psi^+}, \ket{\psi^-}\}$. Similarly, when the second outcome $(\ket{\tilde{\psi}^-})$ clicks on Bob's side, Alice needs to distinguish between two orthogonal states, $\{\ket{\phi^+}, \ket{\psi^+}\}$. Therefore, the unitaries are distinguishable in LDR.
\end{proof}
% The sketch of the proof is as follows. Under GDR, when only a single-system probe is allowed, the evolution produces a set of five states in $\mathbbm{C}^4$. Regardless of the chosen probe, this set of evolved states is trivially indistinguishable. In the local adaptive protocol, the commencing party needs to eliminate at least three unitaries so that she can distinguish two unitaries from the next subsystem. One cannot distinguish more than two unitaries as the evolved states with single system probe will span a two-dimensional subspace. We check all the possible eliminations of each subsystem and show that no measurement can accomplish all the favourable eliminations such that the unitaries are distinguishable with a single system. That confirms the unitaries are indistinguishable in GDR and one-way LDA. For the second part of the proof, Alice and Bob choose Bell state as the probe and use the Bell basis measurement to distinguish the unitaries perfectly in the LOCC protocol. (see End-Matter for details).

\section{Discussion}
In summary, this work shows that nonlocality without entanglement, when extended to unitary processes, exhibits intriguing features and imposes fundamental structural limitations on unitary discrimination across different strategies. Through a series of theorems, we establish non-equivalence and non-hierarchical relations between the adaptive and restricted strategies in both global and LOCC settings. We also further highlight the usefulness of entangled probing states. FIG.~\ref{venn} summarizes our main results and includes the two most immediate and important open problems.

\begin{figure}
    \centering

% --- FILL FIRST (no borders here) ---

\begin{tikzpicture}[scale=0.7, transform shape]

% Outer box
\fill[gray!10] (-3.5,-2.7) rectangle (3.5,2.7);

% GDR
\fill[gray!25] (-2.7,-1.6) rectangle (0.9,1.6);

% LDA
\fill[gray!50] (-0.8,-2.3) rectangle (3,0.4);

% Intersection (darkest)
\fill[gray!75] (-0.6,-1.3) rectangle (0.7,0.2);

% --- NOW DRAW ALL BORDERS ON TOP ---

% Outer box
\draw[gray!70!black, thick] (-3.5,-2.7) rectangle (3.5,2.7);
\node at (0,2.3) {\large GDA};

% GDR
\draw[gray!80!black, thick] (-2.7,-1.6) rectangle (0.9,1.6);
\node at (-1.8,0.3) {\large GDR};

% LDA
\draw[gray!90!black, thick] (-0.8,-2.3) rectangle (3,0.4);
\node at (1.8,-1) {\large LDA};

% Intersection
\draw[black, thick] (-0.6,-1.3) rectangle (0.7,0.2);
\node at (0.09,-0.5) {\large LDR};

\end{tikzpicture}
    \caption{This Venn diagram illustrates Observation \ref{obse_1}, where each rectangle represents the sets of multipartite product unitaries that are distinguishable under the respective strategy. The sets corresponding to GDR and LDA intersect, but neither is contained within the other, as established by Theorems \ref{thr5} and \ref{thr3}. The results lead to two natural open questions: $(i)$ Are GDR and LDA strict subsets of GDA? In other words, does there exist a set of unitaries that is distinguishable under GDA but not under either GDR or LDA? $(ii)$ Is the intersection of GDR and LDA equal to LDR, or are there sets of unitaries that are distinguishable under both GDR and LDA but not under LDR?}
    \label{venn}
\end{figure}

In future, one may attempt to evaluate the success probability of distinguishing a given set of unitaries, which could provide a quantitative measure of the separation between different discrimination strategies.  
Another interesting study can be carried out to find genuinely nonlocal sets of product unitaries or general quantum channels. In the context of quantum states, it is a well-developed field where genuinely nonlocal multipartite quantum systems are globally distinguishable and locally not distinguishable in any partition \cite{halder,zhen,ghosh_s}. Additionally, one can conceive various quantum computation and information-theoretic tasks by applying the results of global and local distinguishability of unitary operations.

\subsection*{Acknowledgment} We gratefully acknowledge the financial support from STARS (STARS/STARS-2/2023-0809), Govt. of India.

\bibliography{ref}

\end{document}